\documentclass[onecolumn, 11pt]{IEEEtran}

\usepackage{cite}
\usepackage{graphicx}
\usepackage[cmex10]{amsmath}
\usepackage{amssymb}
\usepackage{amsfonts,amsthm}

\newtheorem{theorem}{Theorem}

\newtheorem{corollary}{Corollary}
\newtheorem{lemma}{Lemma}
\newtheorem{proposition}{Proposition}
\newtheorem{assumption}{Assumption}
\newtheorem{definition}{Definition}

\begin{document}
\title{Quantization Games on Social Networks and \\ Language Evolution}

\author{Ankur Mani, Lav~R.~Varshney, and Alex (Sandy) Pentland%
\thanks{Portions of the material in this paper were first presented in \cite{ManiVP2013}.  This work
was supported in part by the National Science Foundation under grant CCF-1717530 and a University of Minnesota Digital Technology Initiative Seed Grant.}%
\thanks{A.~Mani is with the Department of Industrial \& Systems Engineering, University of Minnesota, Minneapolis, MN, 55455 USA (e-mail: amani@umn.edu).}%
\thanks{L.~R.~Varshney is with the Coordinated Science Laboratory and the Department of Electrical and Computer Engineering, University of Illinois at Urbana-Champaign, Urbana, IL, 61801 USA (e-mail: varshney@illinois.edu). He is also currently with Salesforce Research, Palo Alto, CA, USA.}%
\thanks{A.~Pentland is with the Media Laboratory, Massachusetts Institute of Technology, Cambridge, MA, 02139 USA (e-mail: pentland@media.mit.edu).}%
}

\maketitle

\begin{abstract}

We consider a strategic network quantizer design setting where agents must balance fidelity in representing
their local source distributions against their ability to successfully communicate with other connected 
agents. We study the problem as a network game and show existence of Nash equilibrium quantizers. For any agent, under Nash equilibrium, the word representing a given partition region is the conditional 
expectation of the mixture of local and social source probability distributions within the region. Since having knowledge of the original source of information in the network may not be realistic, we show that under certain conditions, the agents need not know the source origin and yet still settle on a Nash equilibrium using only the observed sources. Further, the network may converge to equilibrium through a distributed version of the Lloyd-Max algorithm.  In contrast to traditional results in the evolution of language, we find several vocabularies may coexist in the Nash equilibrium, with each individual having exactly one of these vocabularies. The overlap between vocabularies is high for individuals that communicate frequently and have similar local sources.  Finally, we argue that error in translation along a chain of communication 
does not grow if and only if the chain consists of agents with shared vocabulary.  Numerical results are given.
\end{abstract}

\IEEEpeerreviewmaketitle

\newpage

\emph{``Orange in one country is pretty close to orange in another.  One reason for this may be that we see the same sky and the same sun.\dots [However] what some Americans call a large soda some Europeans would call a bucket.''} --- S.~Page \cite[p.~79]{Page2007}

\section{Introduction}
\label{sec:intro}

\IEEEPARstart{W}{hen} representing information, resource constraints often necessitate the quantization of real-valued signals 
into a finite codebook \cite{OrchardB1991,ZegerVG1992, ShlezingerER2019}.  For example, people may use a discrete vocabulary to internally represent the external world 
due to cognitive constraints \cite{SunWGV2012}.  In the traditional approach for performing such quantization, there is a single source
of information to be represented with high fidelity.  Indeed, the apocryphal story of the Eskimoan language 
having a large number of words for snow points to the general principle of \emph{focal vocabularies}: 
specialized sets of distinctions particularly important for a particular focus of experience. 
Adapting language to match local statistics allows a person to represent her environment with greater 
fidelity.

In network settings where several agents interact to represent their local information and also communicate 
information with each other, however, strategic effects arise in the design of quantizers.  The \emph{confusio linguarum}---the 
fragmentation of languages---described in biblical stories, points to the need for \emph{shared vocabularies} 
to enable coordination, collaboration, and exchange in social groups whether to build a city or conduct business.  There needs to be a shared understanding of what words and phrases mean, otherwise groups may become bogged down in ambiguous interpretations of even the most basic of concepts.
Indeed, linguistic distance is associated with less international trade \cite{Lohmann2011}.

Here we study how agents, embedded in a social network, choose a vocabulary to balance interpretation 
loss in communication with others in a social group against loss in individualized representation.  The formalism
we develop to study this problem introduces a 
conflict of interest among agents in the choice of their quantization codebooks and introduces a novel quantization network game.  
Our formulation is quite different from other work in network quantization \cite{FlemingZE2004} in that competition 
among agents is critical and there is no centralized designer of codes, but some local optimality conditions that emerge are similar.  A quantization game for only two agents was developed and studied in \cite{JagerMR2011, OConnor2014, FrankeW2014, LiCalziM2016}, where Voronoi partitions were shown 
to have equilibrium properties; like our approach, the solution concept is based on Nash equilibrium.  
We are only aware of one prior work on game-theoretic quantization with more than two agents \cite{RhimVG2011b}, 
which was concerned with group decision-making in a parallel topology.

Language formation games have also been studied for two agents in the presence of noisy communication channels 
\cite{TouriL2013,HernandezS2014}.  The central results establish generic conditions on when communication may be successful.  
These results can be regarded as contributing to an understanding of signaling games and strategic communication that have been developed in the economics literature (see \cite{CrawfordS1982} and references thereto).  We also consider some level of communication noise.

Besides developing a novel mathematical question in quantizer design with strategic interaction among agents,
our work provides insight into the evolution of language \cite{NowakK1999,NowakKN2002,Niyogi2006,CremerGP2007}. 
Most prior theoretical results point to convergence of language to a single shared vocabulary. 
In contrast we find several vocabularies may coexist under Nash equilibrium
under a variety of settings. Overlap between vocabularies is high for individuals who communicate frequently and have similar 
natural environments, as compared to the agents (like the Eskimos) who communicate less frequently and
have dissimilar local environments. 
This is consistent with the current state of human language where several languages co-exist \cite{Nettle1999},
with the nature of historical language evolution \cite{PagelACM2013}, and with synthetic experimental results \cite{GraesserCK2019}. 

In engineering, there is growing interest in studying the evolution of language for applications in robotics \cite{Steels2011}.  
One approach is through agent-based models, where agents engage in routinized turn-taking 
interactions to develop a common language.  Although agents may be software that operate in virtual worlds, they are primarily 
physical robotic agents that interact with each other in a real world as experienced through a 
sensory-motor system.  Some classical findings are the evolutionary emergence of perceptually grounded 
category lexicons, such as colors \cite{SteelsB2005}.   
In describing how color categorization evolves, three main constraint types are described: \emph{constraints from embodiment}
in the sense of how the sensing apparatus affects what is perceived; \emph{constraints from the environment} in the sense of the 
statistical structure of the environment; and \emph{constraints from culture}, where collective decisions are made by a population.

The formalism we propose mathematizes these qualitative notions and finds best schemes for categorical
communication among agents, whether human or robot.

In our model, agents observe signals in their physical environment as well as in their social environment; their codebooks are grounded 
in their physical and social experiences \cite{RoyP2002,Roy2003}.
Signals in an agent's physical environment are randomly generated from a local continuous source distribution, whereas 
signals in an agent's social environment are received from peers in the network. 
Since agents have finite vocabulary, there are a finite number of possible signals in the 
social environment for each agent. Pairs of peers in the network may have differing frequency of communication. 

Each agent chooses a vocabulary such that distortion due to quantization of signals in 
her physical environment (focal vocabulary) and social environment (shared vocabulary) is minimized. We characterize the equilibrium in which each agent 
draws signals from a mixture of the distributions from her physical and social environments, choosing a vocabulary to minimize loss for this mixture.  
We find that for any agent, under Nash equilibrium, the word that represents a given 
partition region is the conditional expectation of the mixture probability distribution within the region.  In this sense, just a local view on the network is sufficient to develop a good quantization scheme.  

Further characterization is given in the case of cycle-free communication networks among agents. Since the network looks like a forest in this case, for any tree in the forest, the agents may sequentially optimize their vocabularies from root to leaves and still converge to Nash equilibrium. We provide a more general result that says that the agents will converge to the equilibrium irrespective of the sequence in which they optimize their vocabularies and may even use Lloyd-Max iterations instead of full optimization.

Even when individuals have the same codebook, each individual in the social network may have slightly different 
partition regions for the words in the codebook.  This provides an explanation for why the original 
meaning is often lost in translation when communication happens through a long chain of individuals, see Fig.~\ref{fig:loss}.
We find that the error in translation along a chain of communication does not grow if and only if the 
chain consists of agents with shared vocabulary.

\begin{figure}
  \centering
  \includegraphics[width=2.6in]{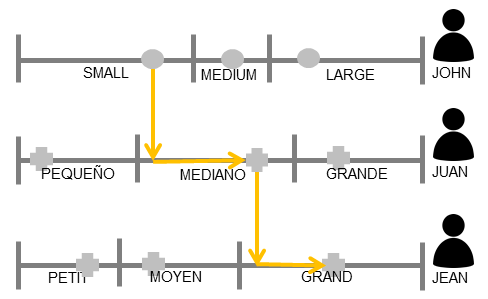}
  \caption{A schematic depiction of the loss in translation phenomenon through a chain of communication, as the word \emph{small}
ends up being interpreted as \emph{grand}.}
  \label{fig:loss}
\end{figure}

Further characterization of the Nash equilibrium shows agents in the network may converge to
equilibrium through a distributed Lloyd-Max algorithm \cite{GershoG1992,Max1960}. 
Dynamics under the
distributed Lloyd-Max algorithm have relations to prior investigations of language evolution \cite{BaronchelliFLCS2006}, 
but we consider the balance of focal and shared vocabularies rather than seeing the convergence to shared
vocabularies.  Numerical examples are given.

The present paper expands on the first presentation of this work \cite{ManiVP2013} by generalizing technical assumptions and clarifying the proofs of results.  It also formalizes results on translation chains and includes numerical results to provide insight.

\section{Problem Statement}
\label{sec:problem}

Consider a set of $N$ agents.  Each agent $i = 1,\ldots,N$ observes signals in an uncertain environment. 
The agents either receive signals directly from their own physical environment or indirectly from other agents via their 
social environments. The random vector $\mathbf{z}_i \in \{0, 1\}^N$ with the constraint that exactly one element of $\mathbf{z}_i$ is $1$ captures the source of the signal for agent $i$. If the signal is coming from agent $j$ then the $j$th element, $z_i^j = 1$ while other elements are zero. On the other hand if the agent $i$ receives the signal from its physical environment then $z_i^i = 1$ while the rest of the elements are zero. The relative frequency of communication between agents is captured by the stochastic communication 
matrix $\mathbf{P}$, where $P_{ij} = \Pr(z^j_i = 1)$ is the relative frequency with which agent $i$ receives a signal from agent $j$.\footnote{The $\mathbf{P}$ 
matrix may be sparse since agents often communicate only with a small set of neighbors due to limited time and the desire to avoid
information overload \cite{MiritelloMLMBRD2013}, but our results do not require sparsity.}

The physical environment of any agent $i$ may be unique and is represented by a random variable $X^{(p)}_i$ that takes 
values in the alphabet $\mathcal{X} = (0,1)$ and has a continuous probability distribution with density 
function $p^{(p)}_i$.  The true environment of agent $i$ is a mixture of its physical environment and social environment and is represented by the random variable 
\[
X^{(sp)}_i = z^i_i X^{(p)} + \sum_{j \neq i}{z^j_i X^{(sp)}_j}\mbox{.}
\]
The density function $p^{(sp)}_i$ of $X^{(sp)}_i$ is a mixture that satisfies the following:
\begin{equation}
p^{(sp)}_i(x) = P_{ii}p^{(p)}_j(x) + \sum_{j \neq i}{P_{ij}p^{(sp)}_i(x)}
\end{equation}
for all $x \in (0,1)$.  One might wonder why signals from different sources are not handled separately, e.g.\ with individualized decoders, but this would be taxing for agents with bounded ability.

Since agents have bounded ability, they operate with a 
vocabulary of at most $M$ words. An agent $i$ maps the observed signals in her environment onto a vocabulary 
$\mathcal{Y}_i = \{y^1_i,\dots,y^M_i\} \subset [0, 1]$ ($y^l_i \leq y^{l+1}_i$ for all $l < M$) through her word map
$q_i : \mathcal{X} \rightarrow \mathcal{Y}_i$. The agent's decision consists of the chosen vocabulary and the word map, i.e.\ a reproduction alphabet and a quantization strategy. Note that the choice of words is in the closed set $[0, 1]$ rather than the open set $(0, 1)$; this is purely for technical reasons and as we will show later, rational agents will never choose $0$ or $1$ as a word.

Agents transmit to other agents using words from their own vocabulary. We assume that the communication is noisy and the transmitting agent's words are distorted by a zero-mean additive noise with absolutely continuous distribution. 

Since the transmitter and receiver may have different vocabularies, the receiver maps the transmitter's words into her own vocabulary through her word map. One might wonder how the receiver can do the word mapping, since in typical studies of quantization systems, abstract indices are transmitted.  Here we think of representation alphabets that are subsets of the real line, and so a real-valued letter from the transmitter's representation alphabet is passed through the mapping $q_i$ to produce a letter in the receiver's representation alphabet, cf.~Fig.~\ref{fig:loss} for a depiction.  This seems reasonable for communicating concepts like distance, size, degree, concentration, color, time duration, and temperature \cite{OConnor2014}, which motivate our study.

The communicated signals and physical observations together form the observed environment $\hat{X}^{(sp)}_i$ with density function $\hat{p}^{(sp)}_i$ for each agent $i = 1,\ldots, N$, see Fig.~\ref{fig:blockdia}.  Given $z^j_i = 1$ and $q_j(\hat{X}^{(sp)}_j) = y^k_j$, $\hat{X}^{(sp)}_i = y^k_j + \eta$ where $\eta$ is a zero mean communication noise. The density of the observed signal $\hat{X}^{(sp)}_i$ given $z^j_i = 1$ and $q_j(\hat{X}^{(sp)}_j) = y^k_j$ is $f^k_{j}$. This $f^k_{j}$ has support on $(0, 1)$ and is continuous with mean $y^k_j$ and standard deviation $\sigma^k_j < \min\{y^k_j, 1 - y^k_j\}$. Thus observed signals in the social environments of agents are noisy versions of true signals (corrupted by both quantization and additive noise).
These observed signals depend on the path of communication to the receiver from the agent that originally observed the signal in her physical 
environment, see Fig.~\ref{fig:network}.
Note that since the observed social signal depends on the path of communication from the agent that observed the physical signal 
to the agent that finally received the signal, the same true signal may be mapped onto more than one word depending on the path. The density $\hat{p}^{(sp)}_i$ is also a mixture.
The probability that agent $i$ uses the word $y^k_i$ is 
\[
p^k_i = \int_{q^{-1}_i(y^k_i)}{\hat{p}^{(sp)_i}(x)dx}
\]
and therefore the density of the observed signal is
\[
\hat{p}^{(sp)}(x) = P_{ii}p^{(p)}_i(x) + \sum_{j \neq i, k \in \{1, \dots, M\}}{P_{ij}p^k_j f^k_{j}(x)}\mbox{.}
\]

\begin{figure}
  \centering
  \includegraphics[width=3.5in]{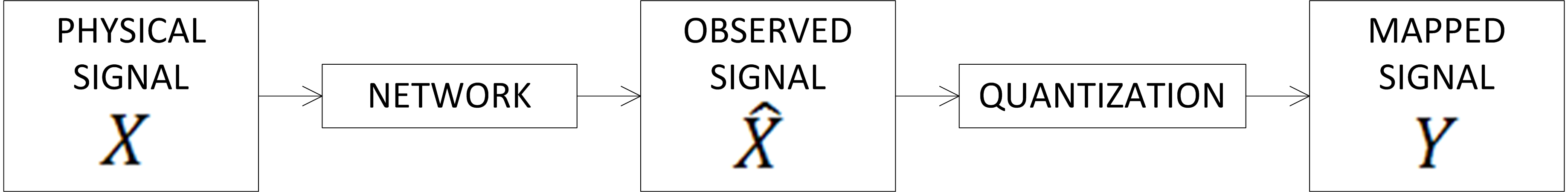}
  \caption{A diagram of how signals are modified as they flow through the network.}
  \label{fig:blockdia}
\end{figure}

\begin{figure}
  \centering
  \includegraphics[width=3in]{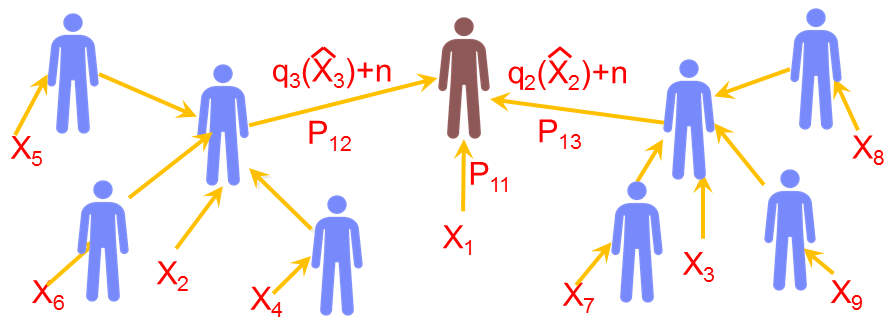}
  \caption{A depiction of how signals take paths through the network.}
  \label{fig:network}
\end{figure}

Since the vocabulary $\mathcal{Y}_i$ is uniquely identified by $q_i$, the strategy of an agent is simply her choice of mapping $q_i$. The  strategy profile $\mathbf{q} = \{q_i\}_{i \in \{1 \dots N\}}$.

The agent's loss due to distortion is measured using $d_i : \mathbb{R}\times\mathbb{R} \rightarrow \mathbb{R}$, where 
$d_i(x, a)$ is strictly convex and coercive in $a$ for all $x$, and symmetric around the true signal $x$. Therefore, the distortion 
is minimized when $a = x$ and strictly increases as $a$ goes farther from $x$. For ease of exposition in the sequel, we assume that the 
distortion is quadratic, i.e.\ $d_i(x, a) = (x - a)^2$. However, all our results hold for any strictly convex, coercive, and symmetric distortion measure, e.g.\ Bregman divergences \cite{BanerjeeMDG2005}.

The loss of risk-neutral agents under the strategy profile $\mathbf{q}$ is the expected distortion in the representation of the true signal:
\begin{align}
\Delta^{(sp)}_i(\mathbf{q}_{-i}, q_i) &= \mathbb{E}\left[d_i\left(X^{(sp)}_i - q_i\left(\hat{X}^{(sp)}_i\right)\right)\right] \\ \notag
&= \mathbb{E}\left[\left(X^{(sp)}_i - q_i\left(\hat{X}^{(sp)}_i\right)\right)^2\right] \mbox{.}
\end{align}
Notice this expected distortion depends on the quantizers of all agents, since $\hat{X}^{(sp)}_i$ is present in the expression. The problem of the risk-neutral agent is to choose $q_i$ that minimizes the above loss. A strategy profile $\mathbf{q}^*$ is a Nash equilibrium if $$q^*_i = \arg\min_{q_i}{\Delta^{(sp)}_i(\mathbf{q}^*_{-i}, q_i)} \mbox{ for all } i \in \{1, \dots, N\}.$$

In this paper, we restrict attention to equilibria in regular quantizers, where partition regions are intervals and words lie within their partition regions. 
\begin{definition}
A quantizer $q_i : \mathcal{X} \rightarrow \mathcal{Y}_i$ for agent $i$ is defined to be \emph{regular} if letters in its representation alphabet $\{y^k_i\}_{k = \{1,2,\ldots,M\}}$ have corresponding partition cells $\{\mathcal{R}^k_i, k = 1, 2, \ldots, M\}$ that are each intervals 
\[
\mathcal{R}^k_i = (a^{k-1}_i,a^k_i] = \{r: a^{k-1}_i < r \le a^k_i\}, k = 1, 2, \ldots, M,
\]
such that for all $k \in \{1, 2, \ldots, M\}$
\[
y^k_i \in (a^{k-1}_i,a^k_i)\mbox{ and } q_i (r) = y^k_i \mbox{ for all } r \in \mathcal{R}^k_i
\]
\end{definition}
We will only consider strategy profiles consisting of regular quantizers for each agent. We will show that such an equilibrium in regular quantizers exists and therefore it is reasonable to study equilibria in regular quantizers and their properties.

We note that the loss of any agent $i$ under a strategy profile $\mathbf{q}$ has three components.
\begin{itemize}
\item The first component, expected \emph{quantization loss} is due to the quantization by agent $i$ and is the expected error in representation of the observed signal:
\begin{align}
\hat{\Delta}^{(sp)}_i(\mathbf{q}_{-i}, \mathbf{q}_i) &= \mathbb{E}\left[d_i\left(\hat{X}^{(sp)}_i - q_i\left(\hat{X}^{(sp)}_i\right)\right)\right] \\ \notag
&= \mathbb{E}\left[\left(\hat{X}^{(sp)}_i - q_i\left(\hat{X}^{(sp)}_i\right)\right)^2\right] \mbox{.}
\end{align}

\item The second component, the expected \emph{communication loss} is due to the noise introduced by the communication of words in the network and is the expected error in the observed signal:
\begin{align}
\tilde{\Delta}^{(sp)}_i(\mathbf{q}_{-i}, \mathbf{q}_i) &= \mathbb{E}\left[d_i\left(X^{\left(sp\right)} - \hat{X}^{(sp)}_i\right)\right] \\ \notag
&= \mathbb{E}\left[\left({X}^{(sp)}_i - \hat{X}^{(sp)}_i\right)^2\right] \mbox{.}
\end{align}

\item The third component is due to correlation between quantization loss and communication loss. In the well-behaved communication networks we consider (see Assumption~\ref{ass:3}), this component will be much smaller and in most cases absent. Therefore, we will ignore this component later.  This is much like mixed distortion in joint source-channel coding that is often zero \cite{KnagenhjelmA1996}.

\end{itemize}

Having specified the players, their possible strategies, and their interlinked payoffs, in the next section we discuss how agents 
choose strategies.  We also make  further technical assumptions.

\begin{assumption}
\label{ass:1}
The semi-elasticity of $p^{(p)}_i(x)$, i.e.\ $\tfrac{\partial \log(p^{(p)}_i(x))}{\partial x}$ is decreasing in $x$ for all $i =1,\ldots, N$.
\end{assumption}
This assumption is equivalent to requiring the density function be log-concave and is widely used in quantization theory to establish uniqueness of locally optimal quantizers \cite{Trushkin1982}.
Many beta distributions and most thin-tailed distributions such as those in the exponential family satisfy this condition.
\begin{assumption}
\label{ass:2}
The semi-elasticity of $p^{(sp)}_i(x)$ and $\hat{p}^{(sp)}$ are decreasing for all choices of strategy profiles $\mathbf{q}$, i.e.\ $\tfrac{\partial \log(p^{(sp)}_i(x))}{\partial x}$ and $\tfrac{\partial \log(\hat{p}^{(sp)}_i(x))}{\partial x}$ are decreasing in $x$ for all $i =1,\ldots, N$.
\end{assumption}
This assumption holds when the communication matrix is diagonal dominant\footnote{Indeed, it is common for people to spend some time interacting with their environment and some time in communication with others \cite{ChuiMBDRSSW2012}, and moreover to have several 
conversation partners with whom they speak regularly \cite{EagleP2006}.} and the agents with high communication probability have similar physical environments.  These assumptions ensures that the loss function is continuous as well as the best response $q^*_i(\mathbf{q}_{-i})$ of agent $i$ to the strategies of other agents $\mathbf{q}_{-i}$ is unique in the strategies of other agents.


\section{Main Results}
\label{sec:results}

The first step in characterizing quantization strategies is to show the existence of a Nash equilibrium in regular quantizers, using continuity and fixed point arguments.
\begin{proposition}
There exists a pure strategy Nash equilibrium in which all agents have regular quantization strategies.
\end{proposition}
\begin{IEEEproof}
We will first show that the loss is continuous in the strategy profile. For each agent $i$ and $n \geq 0$, we define $P^n_i$ as the probability that agent $i$ observes a signal that travelled through a communication path with $n$ cycles containing the agent $i$. Therefore, $P^{n+1}_i = P^n_i P^1_i = \left(P^1_i\right)^{n+1} P^0_i$. The loss can be represented as:
\[
\Delta^{(sp)}_i(\mathbf{q}_{-i}, q_i) = \sum^{\infty}_{n = 0}P^n_i \mathbb{E} \left[ \left(X^{(sp)}_i - q_i\left(\hat{X}^{(sp)}_i\right)\right)^2 | \mbox{ path through } n \mbox{ cycles containing } i \right]
\]

When only agent $i$ changes her vocabulary within an infinitesimal ball of radius $\epsilon$, the change in loss for agent $i$ is also infinitesimal. To observe this, denote the new mapping for $i$ as $q^{\prime}_i$ and the new vocabulary for $i$ as $Y^{\prime}_i$.

The change in expected distortion of signals that come to $i$ for the first time is:
\begin{align*}
|\mbox{err}^0_i| & = \mathbb{E} \left[ \left(X^{(sp)}_i - q^{\prime}_i\left(\hat{X}^{(sp)}_i\right)\right)^2 | \mbox{ path through } 0 \mbox{ cycles containing } i \right]\\
&\qquad - \mathbb{E} \left[ \left(X^{(sp)}_i - q_i\left(\hat{X}^{(sp)}_i\right)\right)^2 | \mbox{ path through } 0 \mbox{ cycles containing } i \right]\\
& \leq O(\epsilon^2) + \sum_{j \neq i}\Pr\left(\|q^{\prime}_i(Y_j + \eta) - q_i(Y_j + \eta )\| \geq \epsilon \right)
\end{align*}
The first term of order $O(\epsilon^2)$ is introduced to all signals due to the shift of the words and to the physical environment signals of agent $i$, at 
the boundaries of the partitions that move to the neighboring partitions. The second term is the probability that the signals in the physical environment 
are mapped onto neighboring partitions. Here the error in the representation of any such signal is bounded by $1$. Clearly, the above error approaches 
$0$ as $\epsilon$ approaches $0$.

The probability of not introducing changes in error due to translation when the signal travels through one cycle containing $i$ is:
\begin{equation}
\rho_i(q^{\prime}_i, q_i, \mathbf{q}_{-i}) \geq \left(\sum_{j \neq i}\Pr\left(\|q_j(Y^{\prime}_i + \eta) - q_j(Y_i + \eta )\| \leq \epsilon \right) \Pr\left(\|q^{\prime}_i(Y_j + \eta ) - q_i(Y_j + \eta )\| \leq \epsilon \right)\right) \mbox{,}
\end{equation}
which approaches $1$ as $\epsilon \rightarrow 0$. For any $n \geq 1$, the contribution to the expected distortion of the signals that come to $i$ through the path with at most $n$ cycles containing $i$ is
\begin{align*}
|\mbox{err}^n_i| &\leq \sum^{n}_{k = 0}\left(P^1_i\right)^{k} P^0_i \left(|\mbox{err}^0_i|  \left(\rho_i(q^{\prime}_i, q_i, \mathbf{q}_{-i})\right)^{k}\right) + \sum^{n}_{k = 0}\left(P^1_i\right)^{k} P^0_i \left(1 -  \left(\rho_i(q^{\prime}_i, q_i, \mathbf{q}_{-i})^{k}\right)\right)\\
& = P^0_i|\mbox{err}^0_i| \frac{\left(1 - \left(P^1_i \rho_i(q^{\prime}_i, q_i, \mathbf{q}_{-i})\right)^{n+1}\right)}{\left(1 - P^1_i \rho_i(q^{\prime}_i, q_i, \mathbf{q}_{-i})\right)} + P^0_i \frac{\left(1 - \left(P^1_i \rho_i(q^{\prime}_i, q_i, \mathbf{q}_{-i})\right)^{n+1}\right)}{\left(1 - P^1_i \rho_i(q^{\prime}_i, q_i, \mathbf{q}_{-i})\right)} - P^0_i \frac{\left(1 - \left(P^1_i \right)^{n+1}\right)}{\left(1 - P^1_i\right)}
\end{align*}
where the first term is for signals in which cycles do not introduce any error due to translation along the path and the second and third terms form the bound on the error introduced due to translation in the $n$ cycles. The total change in the loss for agent $i$ is
\begin{align*}
|\Delta^{(sp)}_i(\mathbf{q}_{-i}, q^{'}_i) - \Delta^{(sp)}_i(\mathbf{q}_{-i}, q_i)| &= \lim_{n \rightarrow \infty}{|\mbox{err}^n_i|}\\
& = P^0_i|\mbox{err}^0_i| \frac{1}{\left(1 - P^1_i \rho_i(q^{\prime}_i, q_i, \mathbf{q}_{-i})\right)} + P^0_i \frac{1}{\left(1 - P^1_i \rho_i(q^{\prime}_i, q_i, \mathbf{q}_{-i})\right)} - P^0_i \frac{1}{\left(1 - P^1_i\right)}
\end{align*}
which approaches $0$ as $\epsilon \rightarrow 0$.

Therefore the loss function of $i$ is continuous in its own strategy. Similarly the loss function of all other agents is continuous in the strategy of $i$. Combining arguments for the finite network, loss functions of all agents are continuous in strategy profile $\mathbf{q}$. Assumptions~\ref{ass:1} and \ref{ass:2} imply that the best response exists and is unique \cite{Trushkin1982}.

Therefore following the theorem of the maximum \cite{Berge1963}, the best response function $\mathbf{q}^*(\mathbf{q})$, where for any $i$, $q^*(\mathbf{q}) = \arg\min_{q^{\prime}_i} \Delta^{(sp)}_i(\mathbf{q}_{-i}, q^{'}_i)$ is non-empty, and continuous. Therefore, following Brouwer's fixed point theorem \cite{Brouwer1911}, a fixed point of the best response function exists.

We also need to show no words belong to $\{0,1\}$ for any fixed point of the best response correspondence for any $i$. This follows  since words do not lie on boundaries of local minima partitions \cite[p.~355]{GershoG1992}. Therefore there exists a Nash equilibrium consisting of regular quantization strategies 
for all agents.
\end{IEEEproof}

For the sequel, we make an assumption that eliminates the correlation between communication loss and quantization loss. In particular, we restrict the set of rational quantizers or the set of quantizers that are not strictly dominated to a well-behaved set.
We first consider the optimal quantizers, $\mathbf{q}^{\circ}$ of all agents in the absence of the communication network, i.e., $$q^{\circ}_i = \arg\min_{q_i}\mathbb{E}\left[d_i\left(X^{(p)}_i - q_i\left({X}^{(p)}_i\right)\right)\right] \\
= \arg\min_{q_i}\mathbb{E}\left[\left(X^{(p)}_i - q_i\left({X}^{(p)}_i\right)\right)^2\right] \mbox{.}$$
\begin{assumption}
\label{ass:3}
We assume the following about the optimal quantizers and the best responses in the network. There exists $\epsilon > 0$ such that
\begin{itemize}
\item for any two agents, $i, j$, with $P_{ij} > 0$ and any $k, l \in \{1, \dots M\}$, the words of $i$ under the quantizer $\mathbf{q}^{\circ}$ are sufficiently far away from the partition boundaries of $j$ in $q^{\circ}_j$, i.e., 
$$\min_{k,l}|y^{\circ k}_i - a^{\circ l}_j| > \epsilon.$$
\item for any agent, $i$, and any $k \in \{1, \dots M\}$ the best response of the agents given any strategy of other agents has limited perturbation from $\mathbf{q}^{\circ}$, i.e., 
$$\max_{k}|y^{\circ k}_i - y^{* k}_i(\mathbf{q})| < \epsilon/2.$$
\item the support of the additive noise in the communication for all agents and all words is less than $\epsilon/2$, i.e.\ $\eta < \epsilon/2$ almost surely.
\end{itemize}
\end{assumption}
These assumptions ensure the zero probability boundary condition 
\cite[p.~355]{GershoG1992} is satisfied and that the interaction between quantization loss and communication loss is zero. These assumptions imply that for any agent $i$, any quantizer $q_i$ with $$\max_{k}|y^{\circ k}_i - y^{ k}_i| \geq \epsilon/2$$ is strictly rationally dominated and the agent $i$  will never choose $q_i$. Therefore, in the sequel we will restrict to the quantization strategies $Q_i = \{q_i : \max_{k}|y^{\circ k}_i - y^{ k}_i| < \epsilon/2\}$ for agent $i$.  This set of quantization strategies are stable with respect to mapping of signals in social environment or are \emph{socially stable}, i.e.\ for any two agents $i,j$ with $P_{i,j} > 0$ and any two indices, $k,l \in \{0, \dots, M\}$, \begin{equation}\label{eq:socially_stable}
q_i(y^l_j) = q^{\circ}_i(y^{\circ l}_j)
\end{equation} for all quantization strategies $q_i \in Q_i$ and $q_j \in Q_j$.   
Under these assumptions, we now characterize the equilibrium. We first show that the best response functions for all agents satisfy the centroid condition with respect to the true signals in the environment.
\begin{lemma}
\label{lem:cent}
A quantizer $q^l_i \in Q_i$ is the local minimum of the loss function for agent $i$, for fixed strategies $\mathbf{q^l}_{-i} \in \mathbf{Q}_{-i}$ of all other agents, if and only if it satisfies the centroid condition: $\mathbb{E}\left[X^{(sp)}_i | q^l_i\left(\hat{X}^{(sp)}_i\right) = y^k_i\right] = y^k_i$.
\end{lemma}
\begin{IEEEproof}
Assumptions~\ref{ass:1} and \ref{ass:2} imply that there is a unique local minimum of the loss function of agent $i$. Assumption~\ref{ass:3} also implies that all quantizers in $Q_i$ are socially stable, see \eqref{eq:socially_stable}. Then since the change in the error is due to the distortion of the agent's own signals, the local centroid condition is satisfied.
\end{IEEEproof}
The immediate consequence of the lemma is that the the centroid conditions are also satisfied in the equilibrium.
\begin{theorem}
\label{thm:cent}
In any Nash equilibrium, the quantizers of all agents satisfy the centroid condition with respect to the distribution 
of the signals in the true environment.
\end{theorem}
\begin{IEEEproof}
From Lemma~\ref{lem:cent}, all local minima of the loss functions and hence the best responses will satisfy the centroid condition. Therefore the Nash equilibrium satisfies it too because each agent's strategy in the Nash equilibrium is the best response to the strategies of all other agents.
\end{IEEEproof}

Using Assumption \ref{ass:3}, we now show that the centroid conditions are also satisfied with respect to the observed signals in the equilibrium.
\begin{corollary}
In any Nash equilibrium, the quantizers of all agents satisfy the centroid condition with respect to the distribution of 
the signals in the observed environment.
\end{corollary}
\begin{IEEEproof}
Pick any Nash equilibrium $\mathbf{q}$. Following Theorem~\ref{thm:cent}, $\mathbb{E}\left[X^{(sp)}_j | q_j(\hat{X}^{(sp)}_j) = y^k_j\right] = y^k_j$, for any agent $j$. Therefore,
\begin{align*}
& \mathbb{E}\left[\hat{X}^{(sp)}_i | q_i(\hat{X}^{(sp)}_i) = y^k_i\right]\\
& = P_{ii} \mathbb{E}\left[X^{(p)}_i | q_i(X^{(p)}_i) = y^k_i\right] + \sum_{j \neq i, l \in \{1,\dots,M\}}{P_{ij}\Pr\left(q_j\left(\hat{X}^{\left(sp\right)}_j\right) = y^l_j\right)\mathbb{E}\left[y^l_j + \eta | q_i(y^l_j + \eta) = y^k_i\right]}\\
& = P_{ii} \mathbb{E}\left[X^{(p)}_i | q_i(X^{(p)}_i) = y^k_i\right] + \sum_{j \neq i, l \in \{1,\dots,M\}}{P_{ij}\Pr\left(q_j\left(\hat{X}^{\left(sp\right)}_j\right) = y^l_j, q_i\left(y^l_j\right) = y^k_i\right)y^l_j}\\
& \mbox{(using Assumption \ref{ass:3} and because } \eta \mbox{ is zero mean)}\\
& = P_{ii} \mathbb{E}\left[X^{(p)}_i | q_i(X^{(p)}_i) = y^k_i\right]\\
&\qquad + \sum_{j \neq i, l \in \{1,\dots,M\}}{P_{ij}\Pr\left(q_j\left(\hat{X}^{\left(sp\right)}_j\right) = y^l_j, q_i\left(y^l_j\right) = y^k_i\right)\mathbb{E}\left[X^{\left(sp\right)}_j | q_j\left(\hat{X}^{\left(sp\right)}_j\right) = y^l_j\right]}\\
& \mbox{(following Theorem \ref{thm:cent})}\\
& = P_{ii} \mathbb{E}\left[X^{(p)}_i | q_i(X^{(p)}_i) = y^k_i\right] + \sum_{j \neq i}{P_{ij}\mathbb{E}\left[X^{\left(sp\right)}_j | q_i\left(q_j\left(\hat{X}^{\left(sp\right)}_j\right) + \eta\right) = y^k_i\right]}\\
& = \mathbb{E}\left[X^{(sp)}_i | q_i(\hat{X}^{(sp)}_i) = y^k_i\right] = y^k_i \mbox{.}
\end{align*}
This completes the proof.
\end{IEEEproof}
This result implies designing for the objective $\Delta^{(sp)}_i(\mathbf{q}_{-i}, q_i)$ and for the objective 
$\hat{\Delta}^{(sp)}_i(\mathbf{q}_{-i}, q_i)$ could be equivalent, thereby paralleling traditional
results in remote source coding where there is separation between estimation and quantization \cite{WolfZ1970}. We now show that
the two objectives are indeed equivalent in cycle-free networks.

\begin{theorem}
\label{thm:NE-eqvl}
When there are no cycles in the communication network, then a strategy profile is a Nash equilibrium if and only if the quantizers minimize the expected distortion in the representation of the observed signal for each agent given the quantizers for all other agents.
\end{theorem}
\begin{IEEEproof}
When there are no cycles in the communication network, then the observed environments of all agents that send messages to any agent $i$ are independent of $i$'s quantizer. The communication frequency matrix imposes a directed acyclic network where the direction of the edges are from transmitters to receivers. In this case,

\begin{align*}
\Delta^{(sp)}_i(\mathbf{q}_{-i}, q_i) &= P_{ii} \mathbb{E}\left[\left(X^{\left(p\right)}_i - q_i\left(X^{\left(p\right)}_i\right)\right)^2\right] + \sum_{j \neq i} P_{ij} \mathbb{E}\left[\left(X^{\left(sp\right)}_j - q_i\left(q_j\left(\hat{X}^{\left(sp\right)}_j\right)\right)\right)^2\right]\\
& = P_{ii} \mathbb{E}\left[\left(X^{\left(p\right)}_i - q_i\left(X^{\left(p\right)}_i\right)\right)^2\right]\\
&\quad + \sum_{j \neq i} P_{ij} \left(\mathbb{E}\left[\left(X^{\left(sp\right)}_j - q_j\left(\hat{X}^{\left(sp\right)}_j\right)\right)^2\right] + \mathbb{E}\left[\left(q_j\left(\hat{X}^{\left(sp\right)}_j\right) - q_i\left(q_j\left(\hat{X}^{\left(sp\right)}_j\right)\right)\right)^2\right]\right)\\
&\quad + 2 \sum_{j \neq i} P_{ij} \left(\mathbb{E}\left[\left(X^{\left(sp\right)}_j - q_j\left(\hat{X}^{\left(sp\right)}_j\right)\right)\right]\mathbb{E}\left[\left(q_j\left(\hat{X}^{\left(sp\right)}_j\right) - q_i\left(q_j\left(\hat{X}^{\left(sp\right)}_j\right)\right)\right)\right]\right)\\
& = \hat{\Delta}^{(sp)}_i\left(\mathbf{q}_{-i}, q_i\right) + \sum_{j \neq i} P_{ij} \left(\mathbb{E}\left[\left(X^{\left(sp\right)}_j - q_j\left(\hat{X}^{\left(sp\right)}_j\right)\right)^2\right]\right)\\
&\quad + 2 \sum_{j \neq i} P_{ij} \left(\left(\mathbb{E}\left[\left(X^{\left(sp\right)}_j - q_j\left(\hat{X}^{\left(sp\right)}_j\right)\right)\right]\mathbb{E}\left[\left(q_j\left(\hat{X}^{\left(sp\right)}_j\right) - q_i\left(q_j\left(\hat{X}^{\left(sp\right)}_j\right)\right)\right)\right]\right)\right)\mbox{.}
\end{align*}

Assume that $\mathbf{q}^*$ is a Nash equilibrium. Then by Theorem~\ref{thm:cent}, 
\[
\mathbb{E}\left[\left(X^{\left(sp\right)}_j - q_j\left(\hat{X}^{\left(sp\right)}_j\right)\right)\right] = 0\mbox{.}
\]
Therefore,
\begin{align*}
& \Delta^{(sp)}_i(\mathbf{q}^*_{-i}, q^*_i) = \hat{\Delta}^{(sp)}_i\left(\mathbf{q}^*_{-i}, q^*_i\right) + \sum_{j \neq i} P_{ij} \left(\mathbb{E}\left[\left(X^{\left(sp\right)}_j - q^*_j\left(\hat{X}^{\left(sp\right)}_j\right)\right)^2\right]\right) \mbox{.}
\end{align*}
Since the second term is independent of agent $i$'s action, $q^*_i$ must minimize $\hat{\Delta}^{(sp)}_i\left(\mathbf{q}^*_{-i}, q^*_i\right)$ because it minimizes $\Delta^{(sp)}_i(\mathbf{q}^*_{-i}, q^*_i)$ for $\mathbf{q}^*_{-i}$.

To prove the other side, assume a strategy profile $\mathbf{q}^*$ is such that $q^*_i = \arg\min_{q_i}{\hat{\Delta}^{(sp)}_i(\mathbf{q}_{-i}, q_i)}$. Without loss of generality, we divide the set of agents in the network as follows. Let the $\{1, \dots, k_0\}$ agents be the ones who only observe their physical environment, i.e.\ $P_{ij} = 0$ for all $i \leq k_0$ and all $j \neq i$. Let $\{k_0+1,\dots,k_1\}$ be the agents who observe social signals only from the agents in $\{1, \dots, k_0\}$ and so on. We say that agents $\{1, \dots, k_0\}$ are upstream from agents $\{k_0+1,\dots,k_1\}$. Then for any agent, $i \in \{1, \dots, k_0\}$, $\Delta^{(sp)}_i(\mathbf{q}^*_{-i}, q^*_i) = \hat{\Delta}^{(sp)}_i\left(\mathbf{q}^*_{-i}, q^*_i\right)$ and $q^*_i$ is the best response of $i$ to $\mathbf{q}^*_{-i}$ and $E\left[\left(X^{\left(sp\right)}_i - q^*_i\left(\hat{X}^{\left(sp\right)}_i\right)\right)\right] = 0$. Now, extending the argument to agents in $\{k_0+1,\dots,k_1\}$, we find that for any agent $i \in \{k_0+1,\dots,k_1\}$, $q^*_i$ will be a best response of $i$ to $q^*_{-i}$ and 
\[
\mathbb{E}\left[\left(X^{\left(sp\right)}_i - q^*_i\left(\hat{X}^{\left(sp\right)}_i\right)\right)\right] = 0 \mbox{.}
\] 
Continuing this way, the same follows for all agents and hence, $\mathbf{q}^*$ is a Nash equilibrium.
\end{IEEEproof}
This result holds when there are no loops in the network and therefore agents do not reflect messages back to the transmitting agents;
it remains to prove the equivalent result for loopy graphs, where intricacies are reminiscent of analyzing loopy belief propagation \cite{IhlerFW2005}.  (See also numerical results in Sec.~\ref{sec:numerical}.)

\subsection{Myopic Dynamics: Distributed Lloyd-Max Algorithm}
Having established equivalence between quantizing based on observed signals rather than true signals, let us 
see whether a dynamic process based only on observed signals can lead to equilibrium strategies.
Indeed, sequential application of the Lloyd-Max algorithm (itself iterative) may lead to Nash equilibrium
for the original game.

\begin{corollary}
\label{cor:LM}
Assume there are no cycles in the communication network. If all agents cyclically use the Lloyd-Max algorithm to minimize the expected distortion in the representation of the observed signal as a response to the quantizers of other agents, then there exist a non-trivial set of initial conditions for which this myopic dynamics will converge to Nash equilibrium.
\end{corollary}
\begin{IEEEproof}
Since the network is acyclic, then if the Lloyd-Max algorithm converges to the global minima every time, the myopic updating of individual strategies of agents $\{1,\dots,k_0\}$ as defined in the construction in Theorem~\ref{thm:NE-eqvl} will always converge to a strategy that minimizes expected distortion in the observed signal in the first iteration. Thence agents in $\{k_0+1,\dots,k_1\}$ will converge to strategies that minimize their expected error in their observed signals as a response to the strategies of all upstream agents in the second iteration. Following this, in finite number of iterations, all agents will converge to the strategies that minimize their expected error in their observed signals as a response to the strategies of all upstream agents in the second iteration. Then the result follows from Theorem~\ref{thm:NE-eqvl}.

Applying a deterministic annealing step \cite{RoseGF1992} allows an initialization of the Lloyd-Max algorithm so that it is guaranteed to converge to the global minimum for any source distribution.  (Computationally simpler approaches to initialization are also often effective \cite{Wu1990} but do not provide theoretical guarantees.)
\end{IEEEproof}

Sec.~\ref{sec:numerical} shows numerical examples for these dynamics.

\subsection{Characterizing the Equilibrium Vocabularies}
Omitting detailed derivation, let us describe equilibrium vocabularies.  
Contrary to extant results in the evolution 
of language that point to the convergence of language to a single shared vocabulary \cite{NowakK1999}, we find 
that under different settings several vocabularies may coexist in equilibrium.  The overlap between vocabularies is high for 
individuals who communicate more frequently and have similar natural environments as compared to the agents who communicate 
less frequently and have dissimilar local environments. Balance must be achieved between the error in local representation 
and social communication.

The loss in translation phenomenon that arises in long chains of communication, recall Fig.~\ref{fig:loss}, is often troubling
since this can lead to large distortions. Further, the loss could be different for different chains of communication. As such, equilibria that lead to path-dependent translation loss could be undesirable. We now characterize the equilibria that introduce path-dependent translation loss and manifest the phenomenon of loss in translation chains. We say that a set of $\mathbb{A}$ agents have a \emph{shared vocabulary} if the intersection of individual 
partition regions have exactly one word for each agent in $\mathbb{A}$, i.e., $$\bigcap_{j\in \mathbb{A}}(a_j^{k-1}, a_j^{k}) \neq \emptyset \mbox{, for all } k \in \{1, \dots, M\} \mbox{, and } $$ $$y_i^k \in \bigcap_{j\in \mathbb{A}}(a_j^{k-1}, a_j^{k}) \mbox{, for all } i \in \mathbb{A}, k \in \{1, \dots, M\}.$$  The loss in translation along a chain of communication does not grow if and only if the chain consists of agents with shared vocabulary.

\begin{theorem}
For any set of agents $\mathbb{A}$, there is no path-dependent translation loss and the loss in translation is bounded by the maximum partition size for all communication chains if and only if all agents in $\mathbb{A}$ have a shared vocabulary.
\end{theorem}
\begin{proof}
The example in Fig.~\ref{fig:loss} provides a counterexample: a set of agents that do not share a vocabulary and therefore the loss in translation is path-dependent and growing in the length of the communication chain. We note that if John had communicated directly to Jean then Jean would have interpreted \emph{moyen} instead of \emph{grand} and had a smaller translation loss. Such examples can be created for any set of agents that do not have a shared vocabulary, demonstrating that in the absence of shared vocabulary, there is path-dependent translation loss. For the other side, we consider a set of agents $\mathbb{A}$ with shared vocabulary. This implies that for any communication chain between $i, j \in \mathbb{A}$, a signal $x_i$ observed by $i$ that is mapped onto a word $y_i^k$ by agent $i$ is also mapped onto the word $y_l^k$ for any agent $l$ along the communication chain from $i$ to $j$ and is mapped onto the word $y_j^k$ by agent $j$. This implies that loss in translation along the communication chain is $|y_j^k - y_i^k| \leq |\min_{l \in A}{a_l^{k}} - \max_{l \in A}{a_l^{k-1}}|$. By the arbitrary choice of the communication chain, the the loss in translation along all communication chains from $i$ to $j$ is $|y_j^k - y_i^k| \leq |\min_{l \in A}{a_l^{k}} - \max_{l \in A}{a_l^{k-1}}|$. By arbitrary choice of the words, this is true for all words and correspondingly the signals represented by those words.

\end{proof}
\section{Numerical Examples}
\label{sec:numerical}

We had described a cyclic and distributed Lloyd-Max algorithm in Corollary~\ref{cor:LM}.  In this section, we demonstrate numerical experiments for an implementation of the algorithm.  These experiments
include loopy networks, and we see convergence even there.

Consider a set of five agents that have the communication matrix 
\[
\mathbf{P} = \begin{bmatrix}
       4/5 & 1/5 & 0 & 0 & 0 \\
       1/10 & 7/10 & 1/10 & 1/10 & 0 \\       
       0 & 1/10 & 4/5 & 1/10 & 0 \\
       1/10 & 1/10 & 1/10 & 3/5 & 1/10 \\
       1/20 & 0 & 0 & 0 & 19/20  
\end{bmatrix}
\]
whose sparsity pattern, agent labels, and agent physical environments are depicted in Fig.~\ref{fig:example1}.  Clearly the network is loopy rather than tree-structured, and agents communicate
with some peers more than others.  
The physical environments for the agents are all governed by beta distributions
with probability density function:
\[
p_i^{(p)}(x) = \frac{1}{B(\alpha,\beta)} x^{\alpha-1} (1-x)^{\beta - 1}
\]
where $B(\cdot,\cdot)$ is the beta function and each agent having different parameters $(\alpha,\beta)$, 
as given in Fig.~\ref{fig:example1}.  One can directly verify that all five sets of parameters yield log-concave density functions, cf.~\cite{Mu2015}. Notice that some pairs of physical environments are more similar to one another than others.

\begin{figure}
  \centering
  \includegraphics[width=5in]{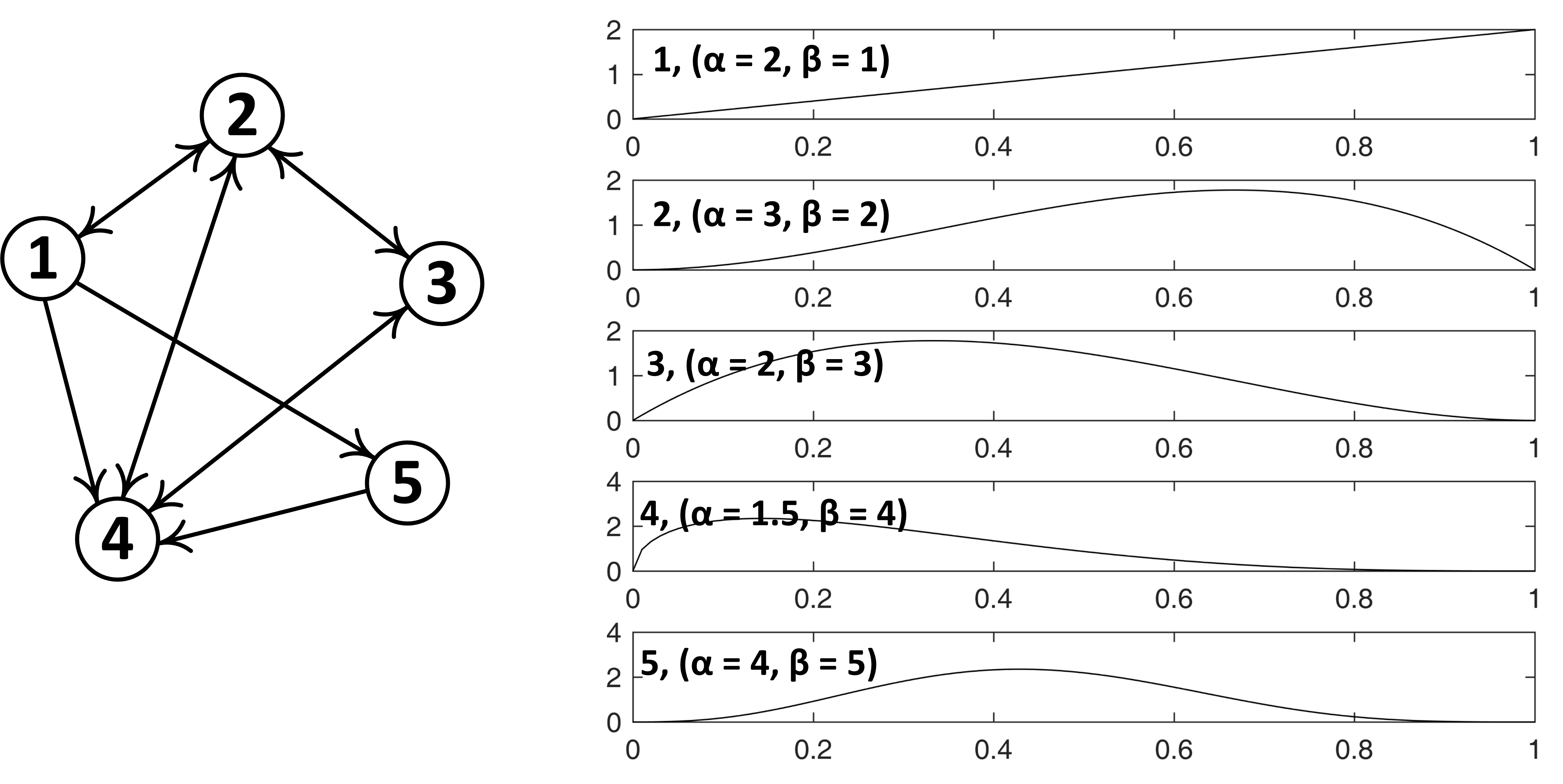}
  \caption{An example network of five agents with communication matrix having sparsity pattern as depicted on the left and local physical environment statistics with beta distribution pdfs with parameters $(\alpha, \beta)$ as depicted on the right.}
  \label{fig:example1}
\end{figure}

We require each agent to use a quantizer with $6$ levels.  The initial quantizer design for each
agent is performed using its respective physical environment $p_i^{(p)}$ alone, with the Lloyd-Max algorithm.  Since these beta distributions are log-concave, the Lloyd-Max algorithm will find the global optimum for any initialization \cite{Trushkin1982}.  In cycling through each $i = 1,\ldots,5$, the next step is to design the agent's quantizer based on the current evaluation of the mixed physical and social density $p_i^{(sp)}$, where the mixing matrix $\mathbf{P}$ is used to combine the physical environment $p_i^{(p)}$ with the current social environment $\hat{p}_j^{(sp)}$ for $j \neq i$ (which is initially just $\hat{p}_j^{(p)}$).  To initialize the Lloyd-Max algorithm, we use the quantizer design from the previous time step, which should already be close to the optimal quantizer.  

This procedure continues cyclically and iteratively with quantizer designs based on new $p_i^{(sp)}$,
which are functions of the fixed $\mathbf{P}$ and $p_i^{(p)}$, but updated $\hat{p}_j^{(sp)}$.  If quantizers do not change too much between a given iteration and the next, we terminate the design and declare convergence.

Fig.~\ref{fig:ex1} shows the quantizers designed in the first four iterations of the algorithm.  For subsequent iterations, the designed quantizers are not easily distinguished visually, since convergence happens quite quickly.  

As one can see, both the physical and the social environments exert an influence on quantizer design, and so the agents' quantizers are more similar to one another in equilibrium than without social considerations.  To see this quantitatively, for each of the $\binom{6}{2}$ pairs of agents we plot the mean-squared distance between the representation points of the quantizers as a function of the Hellinger distance between the physical source distributions.  The Hellinger distance between beta distributions is given by:
\[
1 - \frac{B\left(\tfrac{\alpha_1 + \alpha_2}{2},\tfrac{\beta_1+\beta_2}{2}\right)}{\sqrt{B(\alpha_1,\beta_1)B(\alpha_2,\beta_2)}}\mbox{.}
\]
Fig.~\ref{fig:ex1_r}(a) shows the comparison for quantizers designed in iteration 1, using only the physical environment whereas Fig.~\ref{fig:ex1_r}(b) shows the same comparison in iteration 20 after reaching equilibrium.  One can observe that the similarity of physical environments is highy predictive in the first setting (validating the measure), but much less so in equilibrium.  One can also observe the vertical scale is very different in the two subplots: the quantizers are much more similar in equilibrium.  To specifically see the influence of communication, we can see e.g.\ how agent 1 pulls the quantizer for agent 5 to be finer at larger values in the unit interval, where it has greater probability mass (though it gets pulled to the center by other agents). Notice that agent 1 is the only agent sending messages to agent 5 in the example. The right shift in representation points is shown in Tab.~\ref{tab:shift}.

\begin{table}
  \centering
  \caption{Social Influence on Quantizer}
  \begin{tabular}{|l|l|l|l|l|l|l|} \hline
  Agent 5 & $y_1$ & $y_2$ & $y_3$ & $y_4$ & $y_5$ & $y_6$  \\ \hline
  $X_5^{(p)}$ & $0.1982$ & $0.3243$ & $0.4387$ & $0.5574$ & $0.6902$ & $0.8626$ \\ \hline
  $\hat{X}_5^{(sp)}$ & $0.2002$ & $0.3290$ & $0.4444$ & $0.5617$ & $0.6922$ & $0.8641$ \\ \hline
  right shift & $0.0020$ & $0.0046$ & $0.0058$ & $0.0042$ & $0.0020$ & $0.0015$ \\ \hline
  \end{tabular}
  \label{tab:shift}
\end{table}

\begin{figure}
  \centering
  \includegraphics[width=3.1in]{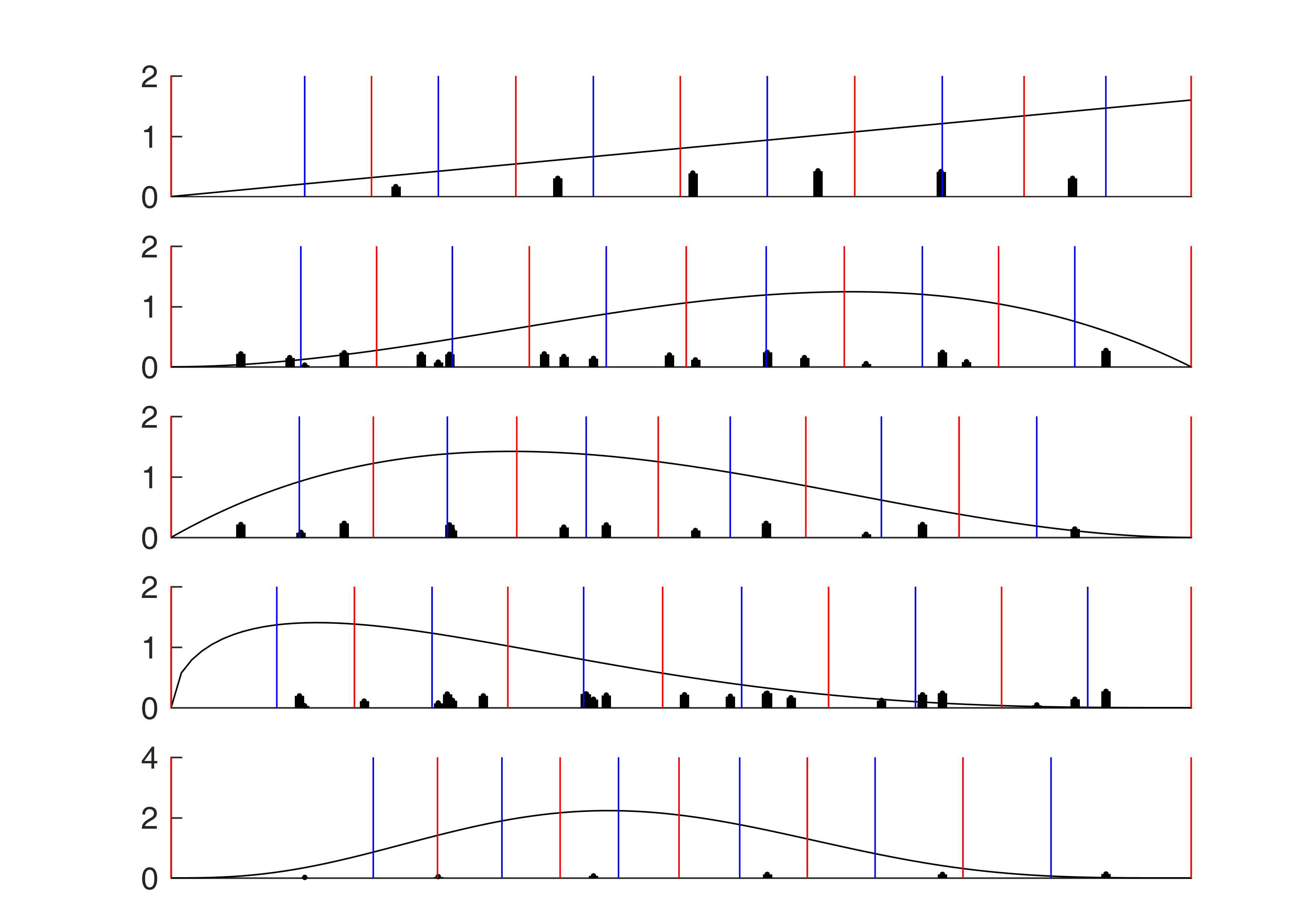} \includegraphics[width=3.1in]{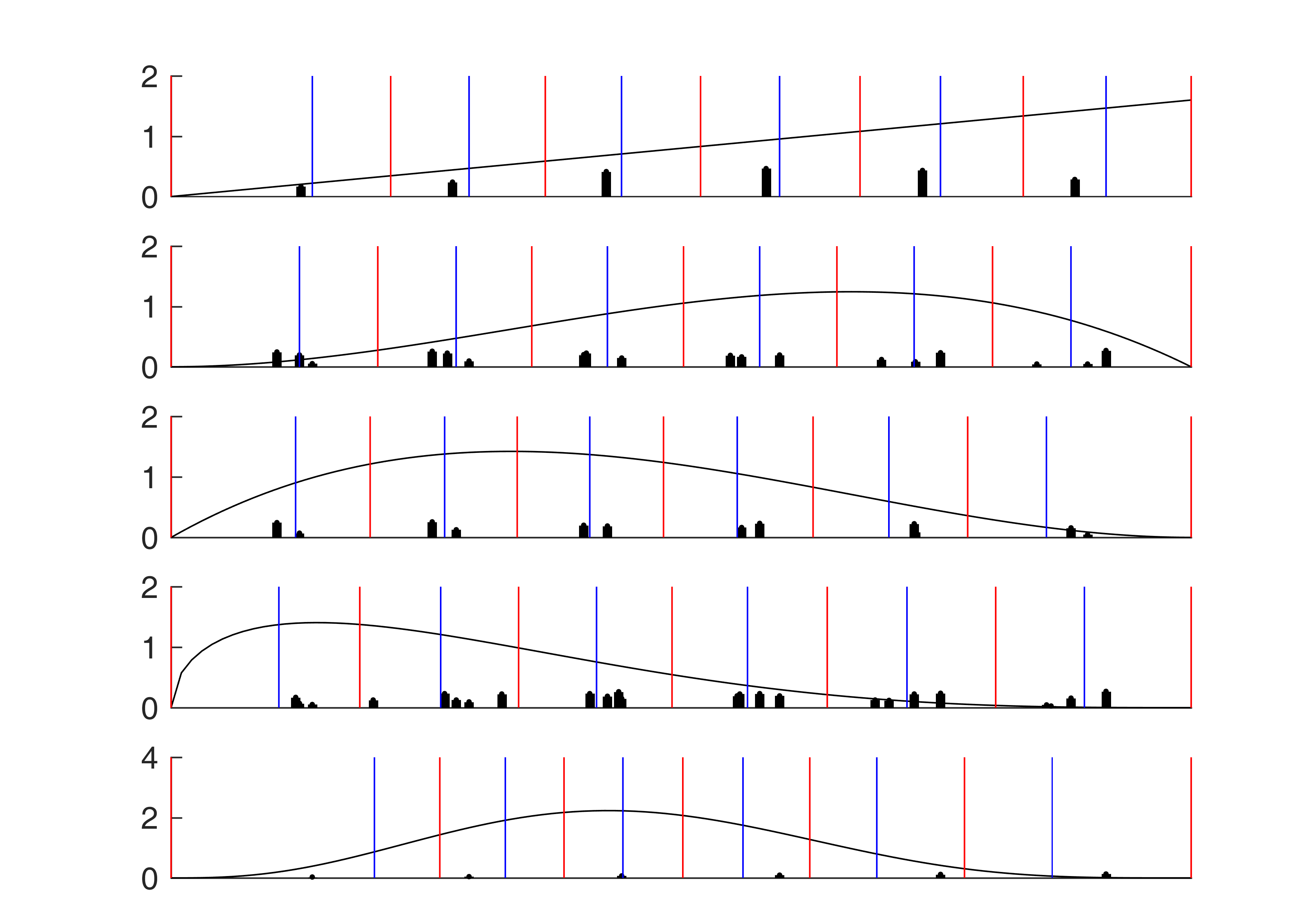}\\ (a)\qquad\qquad\qquad\qquad\qquad\qquad\qquad\qquad\qquad\qquad\quad (b)\\
  \includegraphics[width=3.1in]{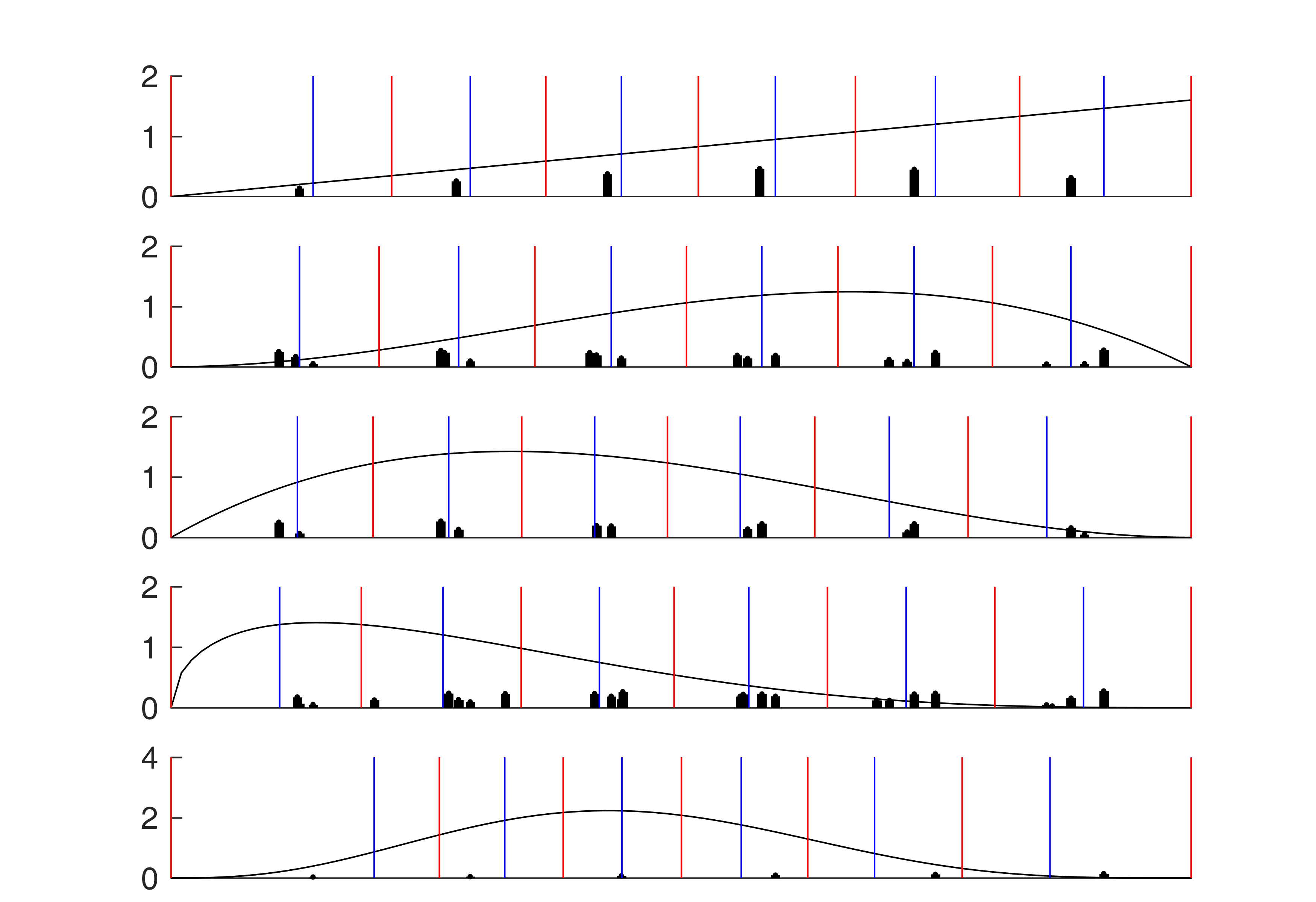} \includegraphics[width=3.1in]{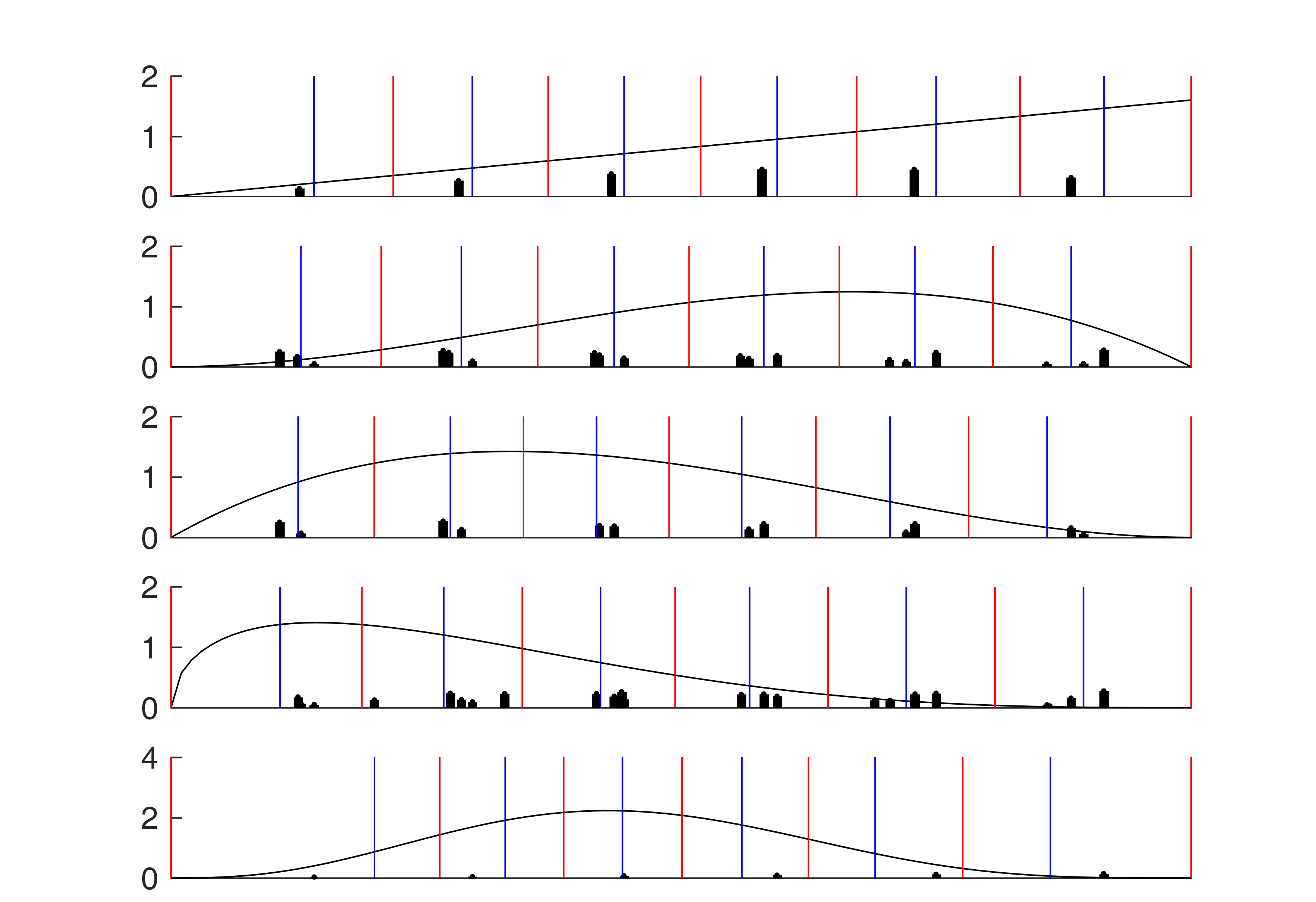}\\ (c)\qquad\qquad\qquad\qquad\qquad\qquad\qquad\qquad\qquad\qquad\quad (d)
  \caption{Iterative quantizer design for the example from Fig.~\ref{fig:example1}.
The black curves indicate the unchanging physical environment statistics, the black vertical stems indicate the changing social context (scaled by 10), the blue lines indicate quantizer representation points, and the red lines indicate quantizer cell boundaries.  Plots (a)--(d) correspond to iterations 1--4, for each of the five agents.}
  \label{fig:ex1}
\end{figure}

\begin{figure}
  \centering
  \includegraphics[width=3.1in]{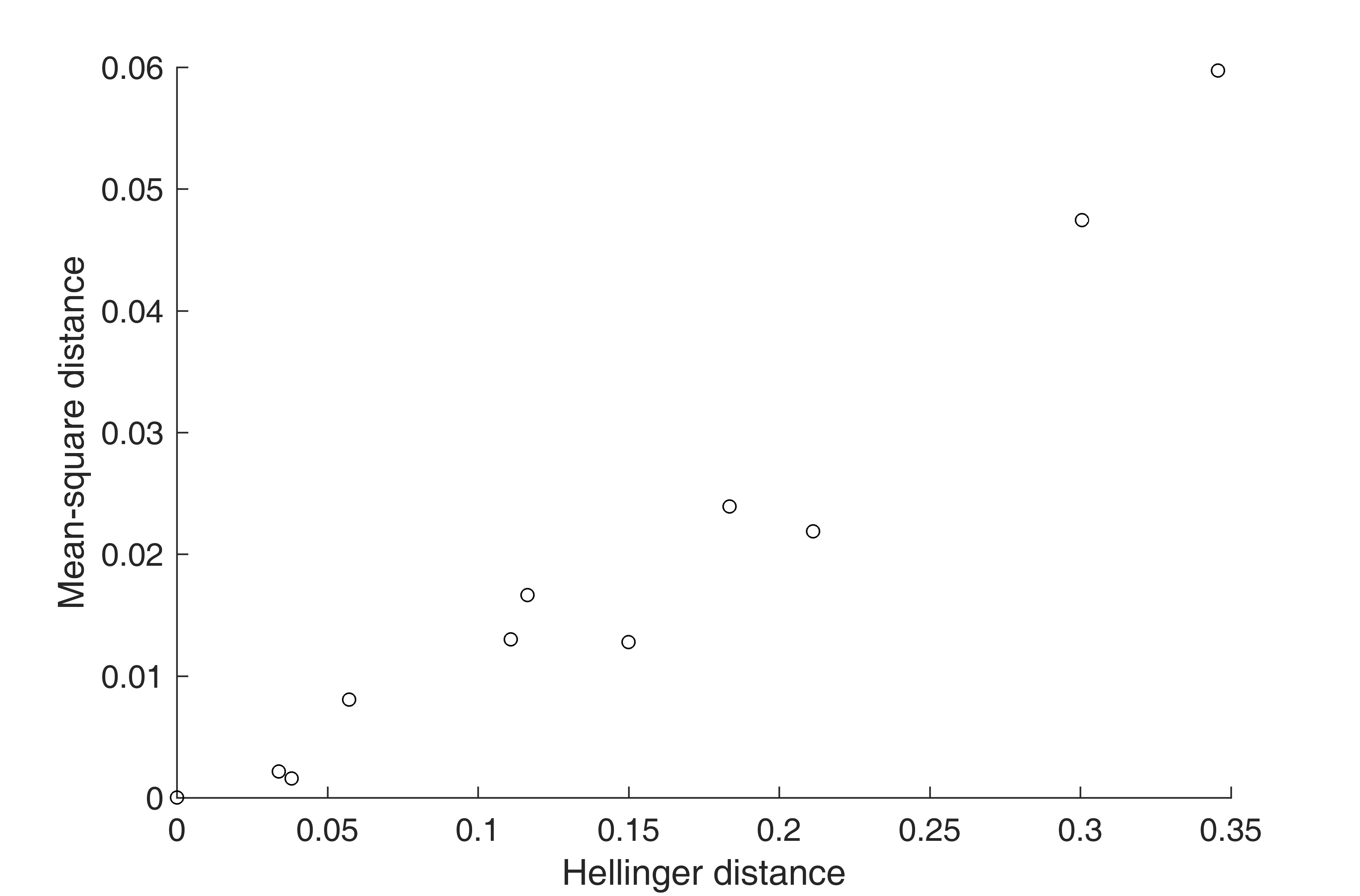} \includegraphics[width=3.1in]{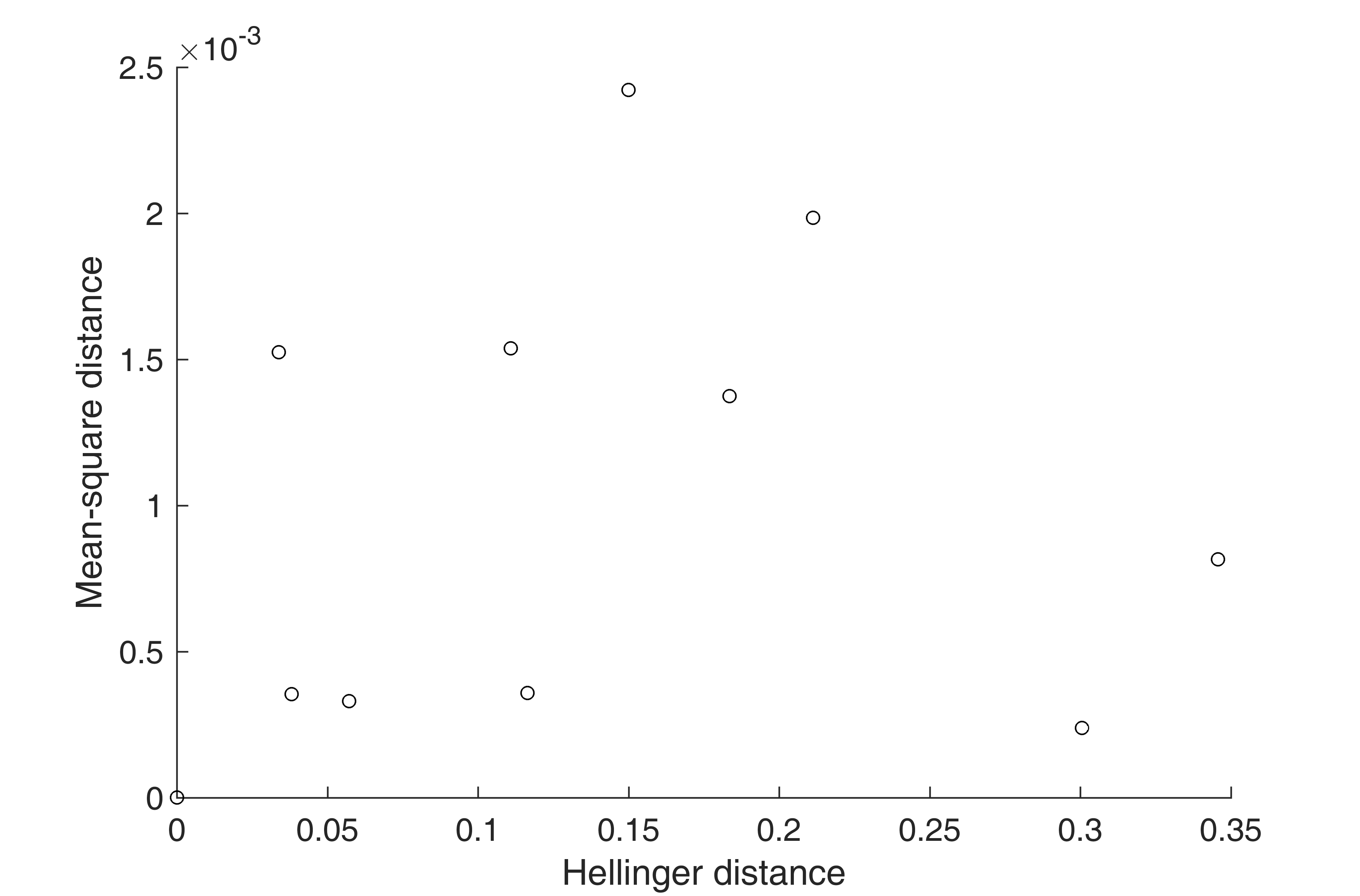}\\  (a)\qquad\qquad\qquad\qquad\qquad\qquad\qquad\qquad\qquad\qquad\quad (b) 
  \caption{Mean-sequare distance between quantizer representation points as a function of the Hellinger distance between physical environment distributions for all pairs of agents in Fig.~\ref{fig:example1}. (a) Quantizer design only considers physical environments. (b) Quantizer design in equilibrium.}
  \label{fig:ex1_r}
\end{figure}

\section{Conclusion}
Symbolic communication requires establishing a vocabulary with a clear association
between the external world and its representation. Unlike traditional biological evolution where organisms
become adapted to their local niches (and unlike traditional quantization theory where codebooks are adapted to
local source distributions), language evolution is necessarily a social phenomenon, since without social interaction 
there is no need for shared vocabularies. In this
work, we have argued that vocabularies evolve to balance individual concerns and social exchange.

As far as we know, we have put forth a first general game-theoretic formulation of quantization theory.
By thinking about the interaction among several connected agents, several novel quantization problems beyond the one
we studied are suggested.

Going forward, we are interested in studying how changes in the communication
frequencies between peers lead to changes in vocabularies and how the birth and death of agents in
the network impact vocabulary evolution. This will help understand the evolution and future of
human languages in a more connected world, and at a smaller scale, the evolution of vocabularies in
collaborative tagging on social media platforms such as del.icio.us \cite{RobuHS2009,Danescu-Niculescu-MizilWJLP2013}. We are also interested in
settings where agents are organizations rather than individuals, and in particular how to ensure
vocabularies evolve to yield clear, concise communication in cross-enterprise collaborations \cite{GulatiWZ2012,VarshneyO2011b}.
Finally, applications in language formation for robotic teams and even human-robot collaboration networks \cite{GleesonMHCA2013}, as well as related questions in unsupervised learning, are 
of interest.

\bibliographystyle{IEEEtran}
\bibliography{abrv,conf_abrv,mani_lib}
\end{document}